\documentclass[11pt]{article}

\usepackage{amsmath,amssymb,amsbsy,amsfonts,amsthm,latexsym,
            amsopn,amstext,amsxtra,euscript,amscd,color,breqn}

\PassOptionsToPackage{hyphens}{url}\usepackage{hyperref}
    
\usepackage[capbesideposition=outside,capbesidesep=quad]{floatrow}

\restylefloat{table}
            
\usepackage{multirow,caption}

\newfont{\teneufm}{eufm10}
\newfont{\seveneufm}{eufm7}
\newfont{\fiveeufm}{eufm5}

\usepackage{systeme}

\usepackage{xcolor}
\usepackage{tikz}
\usetikzlibrary{shapes}
\usetikzlibrary{decorations.markings}

\usepackage{pstricks}
\usetikzlibrary{shapes,arrows}
 \usepackage{pstricks-add}
\usepackage{epsfig}
 \usepackage[space]{grffile} 
 \usepackage{etoolbox} 
 \makeatletter 
 \patchcmd\Gread@eps{\@inputcheck#1 }{\@inputcheck"#1"\relax}{}{}
 \makeatother

\newtheorem{remark}{Remark}

\newtheorem{thm}{Theorem}
\newtheorem{lem}[thm]{Lemma}
\newtheorem{cor}[thm]{Corollary}

\newtheorem{rem}[thm]{Remark}

\newcommand{\Tr}{{\rm Tr}}
\newcommand{\Trn}{{\rm Tr}_n}

\def\+{\oplus}

\def\cN{{\mathcal N}}

\def\F{{\mathbb F}}

\def\00{{\bf 0}}
\def\11{{\bf 1}}
\def\+{\oplus}

\def\\{\cr}
\def\({\left(}
\def\){\right)}

\newcommand{\cardinality}[1]{\# #1}

\usepackage[colorinlistoftodos,prependcaption,textsize=tiny]{todonotes}

\providecommand{\newoperator}[3]{%
  \newcommand*{#1}{\mathop{#2}#3}}

\newoperator{\FD}{\mathrm{FD}}{\nolimits}

\begin{document}
\title{\bf The $c$-differential behavior of the inverse function under the $EA$-equivalence}
\author{\Large Pantelimon~St\u anic\u a, Aaron Geary\\ \\
Applied Mathematics Department, \\
Naval Postgraduate School, Monterey, USA. \\
E-mail: pstanica@nps.edu}

\maketitle

\begin{abstract}
While the classical differential uniformity ($c=1$) is invariant under the CCZ-equivalence, the newly defined~\cite{EFRST20} concept of $c$-differential uniformity in general is not invariant under EA or CCZ-equivalence,   as was observed in~\cite{SPRS20}. In this paper, we find an intriguing behavior of the inverse function, namely, that adding some appropriate linearized monomials increases the $c$-differential uniformity (cDU) significantly, for some~$c$. For example, adding the linearized monomial $x^{2^d}$ to $x^{2^n-2}$, where $d$ is the largest nontrivial divisor of $n$, increases the mentioned $c$-differential uniformity from~$2$ or $3$ (for $c\neq 0,1$) to $\geq 2^{d}+2$, which in the case of the inverse function (as used in the AES) on $\F_{2^8}$ is a significant value of~$18$.  We  consider the case of perturbations via more general linearized polynomials and give bounds for the cDU  based upon character sums. We further provide some computational results on other known Sboxes.
\end{abstract}


\section{Introduction and basic definitions}

The authors of~\cite{BCJW02}  modified/extended the differential attack on some ciphers (for example,  a variant of the IDEA cipher) that use modular multiplication as a primitive operation by using a new type of differential, namely $\left(F(cx),F(x) \right)$ in lieu of $\left(F(x+a),F(x)\right)$.  
Those authors propose that one should look  at other types of differentials for a Boolean (vectorial) function $F$, not only the usual $\left(F(x+a),F(x)\right)$.   

  Inspired by that challenge, we defined in~\cite{EFRST20} a multiplier differential and difference distribution table   (in any characteristic) and later we extended the notion of boomerang connectivity table  in~\cite{S20}. We characterized some of the known perfect nonlinear functions and the inverse function through this new concept.  We also characterized this concept via the Walsh transforms as Li et al.~\cite{Li19} did for the classical boomerang uniformity. Several papers have been written meanwhile on this concept of $c$-differential uniformity (which, unbeknown to us in~\cite{EFRST20}, generalized the recent~\cite{BT} concept of quasi planarity: a quasi planar function   is simply  a PcN function for $c=-1$) .     

 We will introduce here only some needed notation on Boolean (binary, $p=2$) and $p$-ary functions (where $p$ is an odd prime), and the reader can consult~\cite{Bud14,CH1,CH2,CS17,MesnagerBook,Tok15} for more on these objects.

For a positive integer $n$ and $p$ a prime number, we let $\F_{p^n}$ be the  finite field with $p^n$ elements, and $\F_{p^n}^*=\F_{p^n}\setminus\{0\}$ be the multiplicative group (for $a\neq 0$, we often write $\frac{1}{a}$ to mean the inverse of $a$ in the multiplicative group). We let $\F_p^n$ be the $n$-dimensional vector space over $\F_p$.  We use $\cardinality{S}$ to denote the cardinality of a set $S$ and $\bar z$, for the complex conjugate.
We call a function from $\F_{p^n}$ (or $\F_p^n$) to $\F_p$  a {\em $p$-ary  function} on $n$ variables. For positive integers $n$ and $m$, any map $F:\F_{p^n}\to\F_{p^m}$ (or, $\F_p^n\to\F_p^m$)  is called a {\em vectorial $p$-ary  function}, or {\em $(n,m)$-function}. When $m=n$, $F$ can be uniquely represented as a univariate polynomial over $\F_{p^n}$ (using some identification, via a basis, of the finite field with the vector space) of the form
$
F(x)=\sum_{i=0}^{p^n-1} a_i x^i,\ a_i\in\F_{p^n},
$
whose {\em algebraic degree}   is then the largest Hamming weight of the exponents $i$ with $a_i\neq 0$. 
We let $\Trn:\F_{p^n}\to \F_p$ be the absolute trace function, given by $\displaystyle \Trn(x)=\sum_{i=0}^{n-1} x^{p^i}$ (we will denote it by $\Tr$, if the dimension is clear from the context). 
 
Given a $p$-ary  function $f$, the derivative of $f$ with respect to~$a \in \F_{p^n}$ is the $p$-ary  function
$
 D_{a}f(x) =  f(x + a)- f(x), \mbox{ for  all }  x \in \F_{p^n},
$
which can be naturally extended to vectorial $p$-ary functions.

For an $(n,n)$-function $F$, and $a,b\in\F_{p^n}$, we let $\Delta_F(a,b)=\cardinality{\{x\in\F_{p^n} : F(x+a)-F(x)=b\}}$. We call the quantity
$\delta_F=\max\{\Delta_F(a,b)\,:\, a,b\in \F_{p^n}, a\neq 0 \}$ the {\em differential uniformity} of $F$. If $\delta_F= \delta$, then we say that $F$ is differentially $\delta$-uniform. If $\delta=1$, then $F$ is called a {\em perfect nonlinear} ({\em PN}) function, or {\em planar} function. If $\delta=2$, then $F$ is called an {\em almost perfect nonlinear} ({\em APN}) function. It is well known that PN functions do not exist if $p=2$.

 
For a $p$-ary $(n,m)$-function   $F:\F_{p^n}\to \F_{p^m}$, and $c\in\F_{p^m}$, the ({\em multiplicative}) {\em $c$-derivative} of $F$ with respect to~$a \in \F_{p^n}$ is the  function
\[
 _cD_{a}F(x) =  F(x + a)- cF(x), \mbox{ for  all }  x \in \F_{p^n}.
\]

For an $(n,n)$-function $F$, and $a,b\in\F_{p^n}$, we let the entries of the $c$-Difference Distribution Table ($c$-DDT) be defined by ${_c\Delta}_F(a,b)=\cardinality{\{x\in\F_{p^n} : F(x+a)-cF(x)=b\}}$. We call the quantity
\[
\delta_{F,c}=\max\left\{{_c\Delta}_F(a,b)\,|\, a,b\in \F_{p^n}, \text{ and } a\neq 0 \text{ if $c=1$} \right\}\]
the {\em $c$-differential uniformity} of~$F$. If $\delta_{F,c}=\delta$, then we say that $F$ is differentially $(c,\delta)$-uniform (or that $F$ has $c$-uniformity $\delta$, or for short, {\em $F$ is $\delta$-uniform $c$-DDT}). If $\delta=1$, then $F$ is called a {\em perfect $c$-nonlinear} ({\em PcN}) function (certainly, for $c=1$, they only exist for odd characteristic $p$; however, as proven in~\cite{EFRST20}, there exist PcN functions for $p=2$, for all  $c\neq1$). If $\delta=2$, then $F$ is called an {\em almost perfect $c$-nonlinear} ({\em APcN}) function. 
When we need to specify the constant $c$ for which the function is PcN or APcN, then we may use the notation $c$-PN, or $c$-APN.
It is easy to see that if $F$ is an $(n,n)$-function, that is, $F:\F_{p^n}\to\F_{p^n}$, then $F$ is PcN if and only if $_cD_a F$ is a permutation polynomial.

In~\cite{EFRST20,SPRS20,RS20,YZ20} various characterizations of the $c$-differential uniformity were found, and some of the known perfect and almost perfect nonlinear functions have been investigated.  In~\cite{S20}, the concept of boomerang uniformity was extended to $c$-boomerang uniformity and characterized  via Walsh transforms, and some of the known perfect nonlinear and the inverse function in all characteristics was dealt with via the  $c$-boomerang uniformity concept.   

The rest of the paper is organized as follows.  Section~\ref{sec2} gives  several background lemmas needed for the remaining of the paper. Section~\ref{sec3}  investigates  $c$-differential uniformity for an EA-perturbation via a linearized monomial of the inverse function.  Section~\ref{sec4} considers a linearized polynomial perturbation and finds some bounds in terms of characters on the finite field.   Section~\ref{sec5} provides a few computational results on some recognizable cipher Sboxes. Section~\ref{sec6}   concludes the paper.

\section{Some lemmas}
\label{sec2}

 We will be using  throughout Hilbert's Theorem 90 (see~\cite{Bo90}), which states that if $\mathbb{F}\hookrightarrow \mathbb{K}$  is a cyclic Galois extension and $\sigma$ is a generator of the Galois group ${\rm Gal}(\mathbb{K}/\mathbb{F})$, then for $x\in \mathbb{K}$, the relative trace $\displaystyle \Tr_{\mathbb{K}/\mathbb{F}}(x)=\sum_{i=0}^{|{\rm Gal}(\mathbb{K}/\mathbb{F})|-1}\sigma^i(x)=0$ if and only if $x=\sigma(y)-y$, for some $y\in\mathbb{K}$.
We also need the following two lemmas.  
\begin{lem} 
\label{lem10} 
Let $n$ be a positive integer. We have:
\begin{enumerate}
 \item[$(i)$] The equation
$ax^2 + bx + c = 0$, with $a,b, c\in\F_{2^n}$, $ab\neq 0$,
has two solutions in $\F_{2^n}$ if  $\Tr\left(
\frac{ac}{b^2}\right)=0$, and zero solutions, otherwise. If $a\neq 0,b=0$, the solution is unique \textup{(}see~\textup{\cite{BRS67}}\textup{)}.
\item[$(ii)$]  The equation
$ax^2 + bx + c = 0$, with $0\neq a,b\in\F_{p^n}$, $p$ odd,
has (two, respectively, one) solutions in $\F_{p^n}$ if and only if the discriminant $b^2-4ac$ is a (nonzero, respectively, zero) square in $\F_{p^n}$.
\item[$(iii)$] The equation
$x^3 + ax + b = 0$, with $a,b\in\F_{2^n}$, $b\neq 0$, has (denoting by $t_1,t_2$ the roots of $t^2+bt+a^3=0$):
\begin{itemize}
\item[$(i)$] three solutions in $\F_{2^n}$ if and only if $\Tr(a^3/b^2)=\Tr(1)$ and $t_1,t_2$ are cubes in $\F_{2^n}$ for $n$ even, and in $\F_{2^{2n}}$ for $n$ odd;
\item[$(ii)$] a unique solution in $\F_{2^n}$ if and only if $\Tr(a^3/b^2)\neq \Tr(1)$;
\item[$(iii)$] no solutions in $\F_{2^n}$ if and only if $\Tr(a^3/b^2)=\Tr(1)$ and $t_1,t_2$ are not cubes in $\F_{2^n}$ \textup{(}$n$ even\textup{)}, $\F_{2^{2n}}$  \textup{(}$n$ odd\textup{)}.
\end{itemize}
\end{enumerate}
\end{lem}

\begin{lem}[\textup{\cite{EFRST20}}]
\label{lem:gcd}
Let $p,t,n$ be integers greater than or equal to $1$ (we take $t\leq n$, though the result can be shown in general). Then
\begin{align*}
&  \gcd(2^{t}+1,2^n-1)=\frac{2^{\gcd(2t,n)}-1}{2^{\gcd(t,n)}-1},  \text{ and if  $p>2$, then}, \\
& \gcd(p^{t}+1,p^n-1)=2,   \text{ if $\frac{n}{\gcd(n,t)}$  is odd},\\
& \gcd(p^{t}+1,p^n-1)=p^{\gcd(k,n)}+1,\text{ if $\frac{n}{\gcd(n,t)}$ is even}.\end{align*}
Consequently, if either $n$ is odd, or $n\equiv 2\pmod 4$ and $t$ is even,   then $\gcd(2^t+1,2^n-1)=1$ and $\gcd(p^t+1,p^n-1)=2$, if $p>2$.
\end{lem}

\section{The $c$-differential uniformity of some $EA$-perturbed inverse function}
 \label{sec3}
 
 We showed in~\cite{EFRST20} that the inverse function is PcN for $c=0$,  and it is 2 or~3 depending upon the parameter  {$c \neq 1$} (we found precisely those conditions).  
 
 In our main result of this paper we see that performing a simple modification of the inverse function increases significantly the maximum value in its $c$-differential spectrum size. In the following, we take $p$ prime, $n\geq 4$ an integer, and  $1\leq t<n$ an integer such that   $\gcd(t,n)=d\geq 1$, $n\geq 3d$ and  $a^{p^t+1}+1=0$ has a root (and consequently, $\gcd(p^t+1,p^n-1)$ roots) in the  field $\F_{p^n}$; this last condition will always happen if $n/d$ is even (that can be seen since $\gcd(p^t+1,p^n-1)$ divides $\frac{p^n-1}{2}$ under the conditions). If $p=2$, this last condition is superfluous.
 
 \begin{thm}
 \label{main:thm}
 Let $p$ be a prime number, $n\geq 4$, $F(x)=x^{p^n-2}$ be the inverse function on $\F_{p^n}$, and $1\neq c\in\F_{p^n}$. Then, the $c$-differential uniformity, $\delta_{G,c}$, of $G(x)=F(x)+x^{p^t}$ satisfies $p^{\gcd(n,t)}+2\leq \delta_{G,c}\leq p^t+4$, if $p=2$, or if $p>2$ and $\frac{n}{\gcd(n,t)}$ is even\textup{;} and $4\leq \delta_{G,c}\leq p^t+4$, if $p>2$ and $\frac{n}{\gcd(n,t)}$ is odd. 
 \end{thm}
 \begin{proof}
 The $c$-differential uniformity equation for $G$ for $c\in\F_{p^n}$ at $(a,b)\in\F_{p^n}\times\F_{p^n}$ is
 \begin{equation}
 \label{eq:cdiff}
 (x+a)^{p^n-2}+(x+a)^{p^t}-c x^{p^n-2}-cx^{p^t}=b.
 \end{equation}

 We first assume that $a\neq 0$. We consider several cases.
 
 \noindent
 {\em Case $(i)$.} Let $x=0$. Equation~\eqref{eq:cdiff} becomes
 \[
 \frac{1}{a}+ a^{p^t}=b.
 \]
 Thus, for any $a\neq 0$ and $b= \frac{1}{a}+ a^{p^t}$, we have a solution of~\eqref{eq:cdiff}, regardless of the value of~$c$.

  \noindent
 {\em Case $(ii)$.} Let $x=-a$. Equation~\eqref{eq:cdiff} becomes
 \[
 c\left( \frac{1}{a}+ a^{p^t}\right)=b,
 \]
 and we have yet another solution to~\eqref{eq:cdiff}, for $c$ given by the above displayed equation. Surely, if $a$ is such that $a^{p^t+1}+1=0$; there are $\gcd(p^t+1,p^n-1)$ such $a$'s (which, by Lemma~\ref{lem10},  is $\gcd(2^t+1,2^n-1)=\frac{2^{\gcd(2t,n)}-1}{2^{\gcd(t,n)}-1}$ if  $p=2$, and if $p>2$, the number of such $a$'s is  $2$ when $\frac{n}{\gcd(n,t)}$ is odd, and $p^{\gcd(n,t)}+1$, when $\frac{n}{\gcd(n,t)}$ is even, all if $t>0$; when $t=0$, the value of  $\gcd(p^t+1,p^n-1)$ is $1$, respectively, $2$, for $p=2$, respectively, $p>2$), then $b$ must be zero and again $c$ can be taken arbitrary.
 
 We make an observation here: the two solutions from Cases $(i)$ and $(ii)$ cannot be combined unless $b=0,a^{p^t+1}+1=0$ (and arbitrary $c$), or  $b=\frac{1}{a}+ a^{p^t}$ and $c=1$.  
 
  \noindent
 {\em Case $(iii)$.}  Let $x\neq 0,-a$.   Equation~\eqref{eq:cdiff} becomes
 \allowdisplaybreaks
\begin{align}
& \frac{1}{x+a}+(1-c) x^{p^t}-\frac{c}{x}=b-a^{p^t},\text{ that is,}\nonumber\\
&x+(1-c)x^{p^t+1} (x+a)-c(x+a)=(b-a^{p^t})x(x+a),\text{ or},\nonumber\\
&  x^{p^t+2} +a  x^{p^t+1} +\frac{b-a^{p^t}}{c-1} x^2+ \frac{ab+c-a^{p^t+1}-1}{c-1} x+\frac{ac}{c-1}=0.\label{eq:gen_t}
\end{align}
We therefore infer that the maximum number of solutions for the $c$-differential uniformity is $\delta_{G,c}\leq p^t+4$. To get a lower bound, we take  $a=0$, obtaining
\[
x^{p^t+2} +\frac{b}{c-1} x^2+x=0,
\]
with solution $x=0$ and cofactor
\begin{equation}
\label{blu:eq0}
x^{p^t+1} +\frac{b}{c-1} x+1=0.
\end{equation}
Multiplying~\eqref{blu:eq0} by $\left(\frac{b}{c-1}\right)^{p^t+1}$ and relabeling $x\mapsto \frac{1-c}{b} x$ (if $b\neq 0$, otherwise, we look at $x^{p^t+1} +1=0$, which  has $\gcd(p^t+1,p^n-1)\leq p^{\gcd(n,t)}+1$ solutions), we obtain
\begin{equation}
\label{blu:eq}
x^{p^t+1}-Bx+B=0,
\end{equation}
where $B=\left(\frac{b}{c-1}\right)^{p^t+1}$ and we can apply~\cite[Theorem 5.6]{Bluher04}. Using the notations from~\cite{Bluher04}, we let $\F_Q=\F_{p^n}\cap \F_{p^t}=\F_{p^{\gcd(n,t)}}$ (so, $Q=p^{\gcd(n,t)}$), $m=[\F_{p^n}:\F_{Q}]=\frac{n}{\gcd(n,t)}$. From \cite[Theorem 5.6]{Bluher04}, we know that there are $\displaystyle \frac{Q^{m-1}-Q}{Q^2-1}$, $\displaystyle \frac{Q^{m-1}-1}{Q^2-1}$, for $m$ even, respectively odd, values of $B$ such that Equation~\eqref{blu:eq} has $Q+1$ solutions.
Let $T$ be the set of all such $B$. Thus, $|T|=\displaystyle \frac{Q^{m-1}-Q}{Q^2-1}$, for $m$ even, and $|T|=\displaystyle \frac{Q^{m-1}-1}{Q^2-1}$, for $m$ odd.

 To get our claimed lower 
bound, we just need to argue that we can always find $b,c$ such that $B\in T$. 
If $m=\frac{n}{\gcd(n,t)}$ is odd, then $\gcd(p^t+1,p^n-1)=1,2$, for $p=2$, respectively, $p>2$, and we can take any $B\in T$, if $p=2$, and $B=\tilde B^2\in T$ (such a $\tilde B$ does exists, for example, $\tilde B=0$) and a random $c\neq 1$, and $b= (c-1) \tilde B^{\frac{2}{p^t+1}}$, for $p>2$. The number of solutions of~\eqref{blu:eq0} for these parameters is therefore $Q+1$.
If $m=\frac{n}{\gcd(n,t)}$ is even, then $\gcd(p^t+1,p^n-1)=Q+1$. We again use~\cite{Bluher04},  by taking $B=\tilde B^{Q+1}\in T$ (such a $\tilde B$ does exists, for example, $\tilde B=0$) and a random $c\neq 1$, and $b= (c-1) \tilde B^{\frac{Q+1}{p^t+1}}$.  
%

 If $a=0$, the $c$-differential equation becomes $x^{p^n-2}+x^{p^t}=b$. If $b=0$, this equation has $x=0$ as a root and moreover, $x^{p^n-p^t-2}+1=0$, which is equivalent to $x^{p^t+1}+1=0$. We therefore have ${\gcd(p^t+1,p^n-1)}$ solutions.
 If $b\neq 0$, the mentioned equation is equivalent to (multiplying by $x\neq 0$)
 \[
 x^{p^t+1}-b x+1=0.
 \]
 We argued above in Case $(iii)$ that this equation has $Q+1$ solutions. 
  
 In the proof above, for $p=2$, one could use~\textup{\cite{HK08}}, where it was shown that an equation of the form $x^{2^t+1}+x+A=0$ has $Q+1$  zeros for $\frac{Q^{m-1}-1}{Q^2-1}$,  $\frac{Q^{m-1}-Q}{Q^2-1}$, for $m$ odd, respectively, even, values of the parameter $A$. Surely, multiplying~\eqref{blu:eq0} by $\left(\frac{b}{c+1}\right)^{\frac1{2^t}}$ \textup{(}which always exists\textup{)} and performing the substitution $x\mapsto x\, \left(\frac{b}{c+1}\right)^{\frac1{2^t}}$ gives us the equation
$\displaystyle 
x^{2^t+1}+x+\left(\frac{c+1}{b}\right)^{1+\frac{1}{2^t}}=0,
$
and we can apply the same technique as in the proof above, though, the existence of values $b,c$ such that $A=\left(\frac{c+1}{b}\right)^{\frac{2^t+1}{2^t}}$ is not in question anymore for any $A\neq 0$. The theorem is shown.
\end{proof}

\begin{rem}
We could have taken $a=b=0$ from the beginning to get the lower bound, but we wanted to emphasize that there are many other entries in the $c$-DDT table of $G$ lower bounded by $p^{\gcd(n,t)}+2$ (under the mentioned conditions on $n,t$).
\end{rem}

  The following corollary is immediate. It implies that if $n=8$, for the inverse function $F(x)=x^{254}$ (which is one of the components of the Sbox used in AES (Advanced Encryption  {Standard})~\cite{AES}, in addition to an affine transformation), the $c$-differential uniformity of $G(x)=x^{254}+x^{2^4}$ has $\delta_{G,c}\geq 18$, for some $c$ (we confirmed computationally that it is exactly~$18$).
  \begin{cor}
  Let $n\geq 4$, $F(x)=x^{p^n-2}$ be the inverse function on $\F_{p^n}$, and $t\,|\,n$ be the largest divisor of $n$ such that $\frac{n}{\gcd(n,t)}$ is even, and  $G(x)=x^{p^n-2}+x^{p^t}$. Then,   there exists $c$ such that $\delta_{G,c}\geq p^t+2$.
  \end{cor}
\begin{remark}
We will see below that, if $p=2$, in fact, any cube $c\neq 1$ satisfies the conditions of the previous corollary.
\end{remark}

Next, we find some values of $t$ for which the upper bound $ p^t+4$, or the lower bound $p^{\gcd(t,n)}+2$ are attained by $\delta_{G,c}$ for some $c$. We will show that, in fact, this will happen for $p=2$, $t=0$ and $n$ even, respectively, $n$ odd. 
 \begin{thm}
 \label{2nd:thm}
 Let $n\geq 4$, $F(x)=x^{2^n-2}$ be the inverse function on $\F_{2^n}$, and $1\neq c\in\F_{2^n}$. Then, if $n$ is even, the $c$-differential uniformity of $G(x)=F(x)+x$ is $\delta_{G,c}=5$,  for some $c$; if  $n$ is odd, there exists $c$ such that  $\delta_{G,c}=4$. Moreover, if $G(x)=F(x)+x^2$ and $n$ is even, then there exists $c$ such that $\delta_{G,c}=5$; if   $n$ is odd and there exists $a$ such that $\Tr\left(\frac{a^2}{a^2+a+1}\right)=\Tr\left( \frac{a^4}{(a+1)^5}\right)=0$, then $\delta_{G,c}=5$ for some $c$ \textup{(}for example, $c=1+ \frac{1}{(a^3+a^2+1)^{\frac12}}$\textup{)}.
 \end{thm}
\begin{proof}
Below, we will not go through the corresponding Cases $(i)$ and $(ii)$ as in Theorem~\ref{main:thm} since these arguments are independent of $t$, but we will refer to them.

Let $G(x)=F(x)+x$. In that case, we must investigate the equation
\begin{equation}
\label{eq:t=0}
x^3+\left(a+\frac{b+a^{2^t}}{1+c} \right) x^2+ \frac{1+c+ab+a^{2^t+1}}{1+c} x+\frac{ac}{1+c}=0.
\end{equation}
To achieve the maximum $5$ number of solutions $x$, then $b=0$, and $a^{2^t+1}$ for $t=0$ must be equal to $1$ (thus,  $a=1$), rendering
\begin{equation}
\label{eq:t=0_1}
x^3+ \frac{c}{1+c}  x^2+ \frac{c}{1+c} x+\frac{c}{1+c}=0.
\end{equation}
Replacing $y=x+\frac{c}{c+1}$, we get
\begin{equation}
\label{eq:t=0_2}
y^3+\frac{c}{(c+1)^2} y+\frac{c}{(c+1)^2} =0.
\end{equation}
By Lemma~\ref{lem10}, this last equation has three solutions if and only if $c\neq 0$ and $\Tr\left(\frac{c}{\left(c+1 \right)^2} \right)=\Tr(1)$ and the roots $t_1,t_2$ of $t^2+\frac{c}{(c+1)^2} t+\left(\frac{c}{(c+1)^2}\right)^3 =0$ are cubes in $\F_{2^n}$, $\F_{2^{2n}}$, for $n$ even, respectively, $n$ odd. We quickly see that $\Tr\left(\frac{c}{\left(c+1 \right)^2} \right)=\Tr\left(\frac{c+1+1}{\left(c+1 \right)^2} \right)=\Tr\left(\frac{1}{c+1}+\frac{1}{\left(c+1 \right)^2} \right)=0=\Tr(1)$, via Hilbert's Theorem 90. Therefore, this can only be potentially achieved if $n$ is even. 

We would need to argue that for $n$ even  we can always find some $c$, such that the solutions to $t^2+\frac{c}{(c+1)^2} t+\left(\frac{c}{(c+1)^2}\right)^3=0$ are cubes in $\F_{2^n}$.
The roots of this equation can be quickly found to be 
\[
t_1=\frac{c}{(c+1)^3},\ t_2=\frac{c^2}{(c+1)^3}.
\]
We immediately see that if we take $c$ to be a cube, then both of these roots are cubes, and consequently we have three roots for~\eqref{eq:t=0_1}. We need to argue that they are not repeated roots. Since we are working over binary, it is sufficient to check that   the coefficient of $x^2$, namely $\frac{c}{c+1}$ is not a root, which is easy since the left hand side of~\eqref{eq:t=0_1} at  $\frac{c}{c+1}$ is exactly $\frac{c}{c+1}\neq 0$, because $c\neq 0$.  
 
 For $n$ odd, we cannot combine Cases $(i)$ and $(ii)$, but the same argument reveals four solutions for~\eqref{eq:t=0_1} and our first claim is shown.

 Let now $t=1$. If $b=0$, and $a^{2^t+1} =a^3=1$ (for $n$ even, we have two options, either $a=1$, or $a^2+a+1=0$; for $n$ odd, we can only have $a=1$). Equation~\eqref{eq:gen_t} becomes
 \[
 x^{4} +a  x^{3} +\frac{a^2}{1+c} x^2+ \frac{c}{1+c} x+\frac{ac}{1+c}=0.
 \]
 
 If $a=1$, the above equation becomes
 \[
 x^4+x^3+\frac{1}{c+1} x^2+ \frac{c}{1+c} x+\frac{c}{1+c}=0,
 \]
 which can be written as
 \[
 (x^2+x+1)((c+1)x^2+c)=0,
 \]
 therefore, we easily get $3$ roots for the above equation, which combined with the ones from Cases $(i)$ and $(ii)$, renders $5$ altogether, for $n$ even. There are many values of $c$ we can take: for example, for any $x\neq 0,1,a$ not a root of $x^2+x+1$, then we take $c=\frac{x^2}{x^2+1}$.
 Surely, if not both Cases $(i)$ and $(ii)$ hold simultaneously, then we still cannot get more than 5 solutions (we may still get 5 solutions, though).

 If $n$ is odd, then $a=1$ cannot give us more than~$3$~roots (since $x^2+x+1\neq 0$, under $n$ odd), so we assume that $a\neq 1$. Again, under $n$ odd, if $b=0$ and $c=0$, and Equation~\eqref{eq:gen_t} becomes
 \[
 x^4+ax^3+a^2 x^2+(1+a^3) x=0,
 \]
 with solutions $x=0, a+1$, and $(x+a)^2+(x+a)+1=0$, but the last equation cannot hold, for $n$ odd. Next, we take $ \frac{1}{a}+ a^{2^t}=b$ (Case $(i)$), and so,
 \[
x^4+a  x^3+\frac{1}{a(c+1)} x^2+\frac{c}{c+1} x+\frac{ac}{c+1}=0.
 \]
  We will find some values of $a,c$ such that the above polynomial can be factored as
  \[
  x^4+a  x^3+\frac{1}{a(c+1)} x^2+\frac{c}{c+1} x+\frac{ac}{c+1}=(x^2+A x+a)\left(x^2+B x+\frac{c}{c+1}\right).
  \]
 Solving the obtained system, we find that 
 \begin{align*}
 & A=\frac{a^2(c+1)+c}{a(c+1)+c},\ B=\frac{(a+1)c}{a(c+1)+c},\\
& \text{ when } c=\left(\frac{a^3+a^2}{a^3+a^2+1}\right)^{1/2}=1+ \frac{1}{(a^3+a^2+1)^{\frac12}}.
 \end{align*}
 Moreover, each factor in the factorization above has two distinct roots (when $AB\neq0$)  if $\Tr\left(\frac{a}{A^2}\right)=\Tr\left(\frac{a^2}{a^2+a+1}\right)=0$ and $\Tr\left(\frac{c}{(c+1)B^2}\right)=\Tr\left( \frac{a^4}{(a+1)^5}\right)=0$. Under the assumption that there are values of   $a\neq 1$ for $n$ odd such that both of these traces are~$0$ (computation reveals that it always happens, but we have been unable to show that in general), the claim is shown.
  \end{proof}
   
   \begin{rem}
   There are many values of $c$ such that the $c$-differential uniformity is maximum, that is, $\delta_{G,c}=5$, $G(x)=x^{2^n-2}+x^2$\textup{:} for example, for $n$ even, any cube $1\neq c$ in $\F_{2^n}$ will do\textup{;} if $n$ is odd, we gave some examples, and computation reveals that there are additional $c$'s.
   \end{rem}

   \section{Perturbations of the inverse function via  linearized polynomials}
   \label{sec4}
   
   As one of the referees suggested, one might wonder what happens if we perturb the inverse function via some other functions. Surely,  equations over finite fields can be very difficult to handle. However, we can get a general result, albeit not as clean as the ones we have already shown, if we perturb the inverse function via a linearized polynomial.
   
   We will need below some more definitions~\cite{LN97}.
   Let $G$ be the Gauss' sum $\displaystyle G(\psi,\chi)=\sum_{z\in\F_q^*} \psi(z)\chi(z)$, where $\chi,\psi$, are additive, respectively, multiplicative characters of $\F_q$, $q=p^n$. Below, we let $\chi_1(a)=\exp\left(\frac{2\pi i\Trn(a)}{q}\right)$ be the principal additive character, and $\psi_k\left(g^\ell\right)=\exp\left(\frac{2\pi ik\ell}{q-1}\right)$ be the $k$-th multiplicative character of $\F_q$, $0\leq k\leq q-2$. 
   
 We take $s_1,\ldots,s_k$ to be the indices $i$ where $a_i\neq 0$ in the linearized polynomial $L$ and $\delta=\gcd(s_1,\ldots,s_k,n)$; also, $p^{\delta \gamma_\alpha}$ is the number of solutions of $T_n(w)=0$, defined in~\eqref{eq:Tn}; as customary, we denote divisibility by a vertical bar.   The lower bound of the theorem below holds under   the following technical condition (see~\cite[Thms. 1.5 and 1.6]{DM02}):
   \begin{equation}
   \begin{split}
  \label{eq:C2}
 & n=2m,  n/\delta\text{ is even}, 2\delta\, |\, s_i-s_j, 4<p^\delta+1|p^{s_i}+1.
   \end{split}
   \end{equation}
   
   \begin{thm}
   \label{thm:mainthm2}
    Let $p$ be an odd prime number, $n\geq 4$, $F(x)=x^{p^n-2}$ be the inverse function on $\F_{p^n}$, and $1\neq c\in\F_{p^n}$. 
    Let $L(x)=\sum_{i=0}^{n-1} a_i x^{p^i}$ be a linearized polynomial.
    Then, the $c$-differential uniformity, $\delta_{G,c}$, of $G(x)=F(x)+L(x)$ satisfies\textup{:}
\begin{itemize}
\item[$(i)$]  $\delta_{G,c}\leq (pN)^{\frac{n}{2}}$, where $N=\max_{\alpha\in\F_q} \{N_\alpha\}$, and $N_\alpha$ is the number of solutions $w$ to
 \begin{equation}
 \label{eq:Tn}
T_n(w)=2\alpha a_0  w+\sum_{i=1}^{n-1} \left((\alpha a_i)^{p^{-i}} w^{p^i}+ (\alpha a_i)^{p^{-2i}} w^{p^{-i}}\right)=0.
\end{equation}
  In fact, $N_\alpha=p^{\delta \gamma_\alpha}$, for some nonnegative integer $\gamma_\alpha$.
\item[$(ii)$] $\displaystyle \delta_{G,c}\geq \frac{1}{p^n} \sum_{\alpha\in\F_q} \chi_1(\alpha) \mu_\alpha p^{\frac{\delta \gamma_\alpha}{2}}$, under the same conditions of ~\eqref{eq:C2}, where $\gamma_\alpha$ is defined in~$(i)$ and $\mu_\alpha=\pm 1$ is the sign of the Weil sum from Equation~\eqref{eq:weil_mod}. Even more precisely,
\[
\delta_{G,c}\geq    (-1)^{\frac{m}{\delta}} p^{-m} \sum_{\substack{\alpha\in\F_q\\ \cN_\alpha=1}}\chi_1(\alpha) +  (-1)^{\frac{m}{\delta}} p^{-m}   \sum_{\substack{\alpha\in\F_q\\ \cN_\alpha>1}} \chi_1(\alpha) (-1)^{\frac{\gamma_\alpha}{2}}p^{\frac{\delta \gamma_\alpha}{2}}.
\]
\end{itemize}
   \end{thm}
   \begin{proof}
   One might wonder if the argument of Theorem~\ref{main:thm} will go similarly in this case.  We will quickly go through Cases $(i)$ and $(ii)$, but the method of Theorem~\ref{main:thm}  will fail after that. The $c$-differential equation is now
   \begin{equation}
   \label{eq:new1}
   (x+a)^{p^n-2}+L(x+a)-cx^{p^n-2}-c L(x)=b.
   \end{equation}
   If $a\neq 0,x=0$, then Equation~\eqref{eq:new1} becomes $a^{-1}+L(a)=b$. Therefore, for any $c$, and $b=a^{-1}+L(a)$, we have a solution of~\eqref{eq:new1}. If $x=-a\neq 0$, then Equation~\eqref{eq:new1} transforms into $c\left(a^{-1}+L(a)\right)=b$, which gives us one more solution of~\eqref{eq:new1}, when $aL(a)+1=0$ and $b=0$ ($c$ is arbitrary). If $0\neq x\neq -a$, then Equation~\eqref{eq:new1} becomes
   \[
   (1-c) x(x+a)L(x)-(b-L(a))x(x+a)+(1-c) x -ca =0,
   \]
   which has at most $\deg L+2$ solutions. 
   
   Take now $a=0$, and obtain
   \begin{equation}
   \label{eq:new2}
   x\left(xL(x)-\frac{b}{1-c} x+1\right)=0.
   \end{equation}
   Unfortunately, this is the point where the method of Theorem~\ref{main:thm} stops being useful, since we do not have any simple method to find the  number of solutions of an equation involving a more general Dembowski-Ostrom (DO) polynomial. We have, however, found a way to get meaningful results by moving into the world of characters as we did in~\cite{S20_Weil} for the boomerang uniformity and {the extended} $c$-boomerang uniformity. 
   
   As done in~\cite{S20_Weil},  
  the number  $N(b)$  of of solutions $(x_1,\ldots,x_n)\in\F_q^n$ ($b$ is fixed) of an equation $f(x_1,\ldots,x_n)=b$ is
\begin{align*}
\cN(b)
&= \frac{1}{q}\sum_{x_1,\ldots,x_n\in \F_q}\sum_{\alpha\in\F_q} \chi_1\left(\alpha \left( f(x_1,\ldots,x_n)- b\right)\right).
\end{align*}
Thus, in addition to $x=0$, the number of solutions $\cN_{b;c}$ of  Equation~\eqref{eq:new2} is given by 
\[
q\cN_{b;c}= \sum_{\alpha\in \F_q} \sum_{x\in \F_q}\chi_1\left(\alpha(xL(x) + \frac{b}{c-1} x+1)\right).
\]
We can perhaps treat {the} $b\neq 0$ case (under some conditions -- and we already did consider such instances in Theorem~\ref{main:thm}, for a monomial $L$), as well, but for our purpose we do not need to, since we are interested in the maximum value in the $c$-DDT, so we simply take $b=0$. The equation above becomes
\[
q\cN_{b;c}= \sum_{\alpha\in \F_q} \chi_1(\alpha) \sum_{x\in \F_q}\chi_1\left(\alpha xL(x) \right).
\]
We now use~\cite{DM02} where it was shown that the Weil sum of a DO polynomial $\sum_{i=0}^{n-1} a_i x^{p^i+1}$ (we simplify his notations, taking $\alpha_i=0,\beta_i=0$, to match our  DO polynomial $xL(x)$) is
\begin{equation}
\label{eq:weil_mod}
\left|\sum_{x\in\F_q} \chi_1\left(\sum_{i=0}^{n-1} a_i x^{p^i+1}  \right)\right|=\sqrt{qN_\alpha},
\end{equation}
where $N_\alpha$ is the number of solutions for $T_n(w)=0$, where
\[
T_n(w)=2A_0  w+\sum_{i=1}^{n-1} \left(A_i w^{p^i}+(A_iw)^{p^{n-i}}\right),
\]
with  $A_i =(\alpha a_i)^{p^{n-i}}$. 
In fact, more precisely, given $\epsilon=\gcd_{1\leq i\leq n-1} \{2s_0,s_0+s_i,s_0+n-s_i,n \}$, the number of solutions to $T_n(w)=0$ is $p^{\gamma_\alpha\epsilon}$, for some nonnegative integer $\gamma_\alpha$.
Applying this to our DO polynomial $\alpha xL(x)$ we get we first claim, that is,
$\delta_{G,c}\leq \sqrt{qN}$, $N=\max_{\alpha\in\F_q}\{ N_\alpha\}$.

We now assume that  $L(x)=\sum_{i=0}^{n-1} a_i x^{p^i}$ and conditions~\eqref{eq:C2} hold.
It was also shown in~\cite[Theorem 1.5]{DM02}  that if $n$ is even, $\delta=\gcd(s_1,\ldots,s_k)$ and $p^\delta+1=p+1$ divides $(p^i+1)$, then the above Weil sum is real and consequently, it is equal to $\mu_\alpha \sqrt{p^{n+\gamma}}$, where $\mu_\alpha=\pm 1$ and $\gamma$ is a nonnegative integer. Therefore, 
\[
\delta_{G,c}\geq \frac{1}{p^n} \sum_{\alpha\in\F_q} \chi_1(\alpha) \mu_\alpha p^{\frac{\gamma_\alpha}{2}}.
\]
We can be more precise and describe $\mu_\alpha$.  Let $S_\alpha=\sum_{x\in\F_q} \chi_1\left(\sum_{i=0}^{n-1} a_i x^{p^i+1}  \right)$. By~\cite[Theorem 1.6]{DM02}, for every $\alpha$, if $ \gamma_\alpha=0$, then $S_\alpha=(-1)^{\frac{m}{\delta}} p^m$, and if $ \gamma_\alpha>0$, then $\gamma_\alpha$ is even and $S_\alpha=(-1)^{\frac{m}{\delta}+\frac{\gamma_\alpha}{2}} p^{m+\frac{\delta \gamma_\alpha}{2}}$.
Thus,
\[
\delta_{G,c}\geq    (-1)^{\frac{m}{\delta}} p^{-m} \sum_{\alpha\in\F_q,\cN_\alpha=1}\chi_1(\alpha) +  (-1)^{\frac{m}{\delta}} p^{-m}\sum_{\alpha\in\F_q,\cN_\alpha>1}\chi_1(\alpha) (-1)^{\frac{\gamma_\alpha}{2}}p^{\frac{\delta \gamma_\alpha}{2}}.
\]
Our theorem is shown.
      \end{proof}

 \section{Some computations on other Sboxes}
 \label{sec5}
 Using SageMath, we recovered the univariate representation of several known cipher Sboxes and ran some computations to see how the $c$-differential uniformity ($c$DU) would change under linearized monomial perturbations.  The ``$c$DU'' column gives the maximum value for $c\neq 1$ and the ``$c$DU w/ linearized monomial'' column represents the maximum value when a monomial of the form $x^{2^i}$ for $0\leq i \leq n-1$ is added to the univariate Sbox polynomial. 
 \begin{center}
\begin{tabular}{ |c|c|c|c| } 
 \hline
 Sbox ($n$-bits) & DU & $c$DU & $c$DU w/ linearized monomial \\ 
 \hline
 Rectangle (4)& 4 & 5 & 7 \\ 
 Serpent-3 (4) & 4 & 6 & 5 \\
 APN (6) & 2 & 6 & 9 \\ 
 Fides (6) & 2 & 7 & 7 \\
 AES (8) & 4 & 9 & 9 \\
 Skipjack (8) & 12 & 8 & 9 \\
 \hline
\end{tabular}
\end{center}

While  {perhaps }not as {pronounced} as Theorem~\ref{main:thm} on the inverse function, the results are interesting in several ways.  In some cases adding a linearized monomial can increase the $c$-differential uniformity to values more than 4 times the regular differential uniformity (see the ``DU'' column) (e.g., to a nontrivial 9/64 bits in APN 6).  In others, adding a monomial keeps the $c$-differential uniformity the same (e.g., Fides 6) or even drops (e.g., Sbox \#3 of Serpent).  
 
  \section{Concluding remarks}
\label{sec6}

In this paper we investigate the $c$-differential  uniformity of some inverse EA-equivalent functions. We show that their $c$-differential uniformity spectrum tends to increase significantly for some $c$, which we believe is not a desirable feature as it indicates some degree of non-randomness.  We also consider arbitrary  linearized polynomials and provide some bounds based upon character sums. We also  provide some computations on some known cipher Sboxes.  Surely, it would be interesting to continue this investigation into more general affine transformations or to consider other functions under EA perturbations and investigate their $c$-differential uniformity.  
    
 \vskip.3cm
\noindent
{\bf Acknowledgments}. The authors   would like to express their sincere appreciation for the reviewers’ careful reading, beneficial comments and suggestions, and to the editors for the prompt handling of our paper.

 \end{document}